\documentclass[runningheads]{llncs}
\usepackage[T1]{fontenc}   
\usepackage{amsmath}
\usepackage{graphicx}
\newtheorem{fact}{Fact}
\newcommand{\junk}[1]{}

\begin{document}
\title{Boolean Matrix Multiplication for
  Highly Clustered Data on the Congested Clique}
\titlerunning{Boolean Matrix Multiplication...}
\author{
  Andrzej Lingas
  \inst{}
}
\authorrunning{A. Lingas}
\institute{
Department of Computer Science, Lund University, Lund, Sweden. 
\email{Andrzej.Lingas@cs.lth.se}}

\maketitle

\begin{abstract}
  We present a protocol for the Boolean matrix product of two
  $n\times b$ Boolean matrices on the congested clique designed for
  the situation when the rows of the first matrix or the columns of
  the second matrix are highly clustered in the space $\{0,1\}^n.$
  With high probability (w.h.p), it
  uses $\tilde{O}\left(\sqrt {\frac M n+1}\right)$ rounds
  on the
  congested clique with $n$ nodes, where $M$ is the minimum of the
  cost of a minimum spanning tree (MST) of the rows
  of the first input matrix and the cost of
  an MST of the columns of the second input matrix in the Hamming
  space $\{0,1\}^n.$ A key step in our protocol is the computation of
  an approximate minimum spanning tree of a set of $n$ points in the
  space $\{0,1\}^n$.  We provide a protocol for this problem (of
  interest in its own rights) based on a known randomized
  technique of dimension reduction in Hamming spaces. W.h.p., it constructs
  an $O(1)$-factor approximation of an MST of $n$ points in the
  Hamming space $\{ 0,\ 1\}^n$ using $O(\log^3 n)$ rounds
  on the congested clique with $n$ nodes.
\end{abstract}
\begin{keywords}
Boolean matrix product, Hamming space, minimum spanning tree (MST),
congested clique.
\end{keywords}

\section{Introduction}
In the distributed communication/computation model of CONGEST, the
focus is on the cost of communication and the cost of local
computation is more or less ignored, in contrast to the classic model
of Parallel Random Access Machine (PRAM) having the opposite focus.
Initially, each processing unit holds a part of the input (typically,
an approximately equal size part) and it is aware of the identities of
the neighboring processing units (nodes) within the underlying
communication network.  The communication and computation proceeds in
rounds.  In each round, each of the processing units can send a
message of logarithmic in the input size number of bits to each other
neighboring processing within the underlying communication
network. In the more powerful {\em unicast} variant, the messages sent by
the unit during a single round can differ, whereas in the weaker
{\em broadcast} variant they must be beidentical.  Each processing unit can
also perform unlimited local computation in each round. The main
objective is to minimize the number of rounds necessary to solve a
given problem on the provided input in this model.

Lotker \emph{et al.} introduced the model of congested clique in in
2005 \cite{LP-SPP05}.  They study the graph minimum spanning tree
(MST) problem in the unicast CONGEST model under the assumption that the
underlying communication network is a clique. Later papers refer to this variant
of CONGEST as congested clique.
\junk{
The distributed communication/computation model of congested clique
has been introduced by
Lotker \emph{et al.} in
2005 \cite{LP-SPP05}.
It focuses on the cost of communication
and ignores that of local computation in contrast  to
the classic model of Parallel Random Access Machine (PRAM) having
the opposite focus.}

The choice of a clique as the underlying communication network is very
natural due to its direct-connectivity and symmetry properties.

Since two decades, efficient protocols for basic (dense) graph
problems have been designed in the congested clique model under the
following assumptions.  Each node of the clique network initially
represents a distinct vertex of the input graph and is aware of
vertex's neighborhood in the input graph.
In each round, each of the $n$ nodes can send a distinct
message of $O(\log n)$ bits to each other node and can perform
unlimited local computation. Note here that when the input graph is
dense, gathering the information about all vertex neighborhoods in a
distinguished node in order to use its unlimited computational power
would require $\Omega(n)$ rounds.  Many graph problems have been shown
to admit substantially faster (with respect to the number of rounds)
protocols \cite{R}, e.g., the minimum spanning tree problem
(MST) admits even an $O(1)$-round protocol on the congested clique
\cite{K_21}.

Also other combinatorial and geometric problems that cannot
be classified as graph ones have been studied in the congested clique
model, e.g., matrix multiplication \cite{CK19}, sorting and routing \cite{L},
and the construction of convex hulls and triangulations
for planar point sets \cite{JLLP24}.
In all these cases,
the basic items, i.e., matrix entries or keys, or points
in the plane,
respectively, are assumed to have $O(\log n)$-bit representations and each node
initially has a batch of $n$ such items. 
As in the graph case,
each node can send a distinct $O(\log n)$-bit message to 
each other node and perform unlimited computation 
in every round. 
In particular, it has been shown that matrix multiplication
(hence, also  Boolean matrix product) admits an
$O(n^{1-2/\omega+\epsilon })$-round algorithm \cite{CK19}, where $\omega$ is the
exponent of fast matrix multiplication (see Preliminaries) 
and $\epsilon $
is any positive constant. Thus, the
matrix multiplication protocol from 
\cite{CK19}
uses roughly at most $n^{0.157}$ rounds
by the current bound $\omega \le 2.371552$ \cite{VXZ24}.
On the other hand, sorting and routing
have been shown to admit $O(1)$-round algorithms \cite{L} under the aforementioned
assumptions.

Augustine et al. presented in \cite{AKP21} several upper bounds on the
number of required rounds for various combinatorial and geometric
problems in the so called $k$-machine mode which includes the congested
$k$-clique model as a special case.
They obtained their bounds by direct PRAM simulations under the
assumptions that the input items are uniformly at random distributed
among the $k$ machines. In case of $n\times $ matrix multiplication,
their upper bound for $k=n$ however is worse than that in \cite{CK19},
and requires the additional assumptions on the input distribution.

Observe that the related problem of computing the matrix product of
  $n\times n$ matrix with an $n$-dimensional column vector admits a
  straightforward $O(1)$-round protocol on the congested clique with
  $n$ nodes. Simply, the node holding the column vector sends its
  $i$-th coordinate to the $i$-th node, for $i=1,...,n,$ in the first
  round.  In the second round, each node sends the received coordinate
  to all other modes.  Consequently, each node can locally compute the
  inner product of its row of the input matrix with the column vector
  in the following round.  However, this observation does not yield
  any $n^{o(1)}$-round protocol for $n\times n$ matrix multiplication
  on the congested clique with $n$ nodes.

The absence of $n^{o(1)}$ upper bounds on
the number of rounds required by matrix multiplication justifies
studying special cases of this problem in the congested clique model.
For instance, protocols requiring fewer rounds for sparse matrices
than that in \cite{CK19} have been presented in \cite{CLT20}.

In the sequential case, among other things, an alternative randomized
approach to Boolean matrix product has been developed in
\cite{BL,GL03}.  The alternative algorithm
computes the Boolean matrix product of two
$n\times n$ Boolean matrices $A$ and $B$ in time $\tilde{O}(n(n+ M
)),$ where $M$ is the minimum of the cost of an MST of the rows of $A$
and the cost of an MST of the columns of $B$ in the Hamming space
$\{0,\ 1\}^n$ \cite{BL}.  The alternative approach is superior solely
when the rows of $A$ or the columns of $B$ are highly clustered, i.e.,
when $M$ is substantially smaller than $n^{\omega -1}.$

In this paper, we reconsider this
approach in the congested clique model.

Our approach to the Boolean matrix product of two $n\times n$ Boolean
matrices on the congested clique raises the question whether there
exists an efficient (with respect to the number of rounds) protocol
for an approximate MST of a set $S$ of $n$ points (vectors) in the
Hamming space $\{ 0,\ 1 \}^n,$ when each node of the congested clique
holds a distinct input point.  The nodes initially are not aware of
the Hamming distances between their points and the other points so the
fast protocols for MST in edge weighted graphs \cite{K_21,R} cannot be
applied directly.  The question is also of an independent interest
since computing an MST is a widely used technique for clustering data,
especially in machine learning \cite{GZJ06,LRN09}.

\junk{We can compute all pairwise
Hamming distances by applying doubly the round protocol for matrix
multiplication over a ring on the congested clique from \cite{CK19}
and then running the $O(1)$-round protocol for MST in edge weighted
graphs from \cite{K_21}. All this yields an $O(n^{1-2/\omega +\epsilon
})$-round protocol for the Hamming MST, for any positive constant
$\epsilon $ (Theorem \ref{theo: exact} in Section 3).

In a number of applications of MST in the Hamming space, a reasonable
approximation of MST is sufficient.  An example of such a application
is Boolean matrix multiplication for highly clustered data
\cite{BL,GL03}.  So the question arises if protocols for a reasonable
approximation of MST in the Hamming space $\{0,\ 0\}^n$ using a small
number of rounds (e.g., a polylogarithmic in $n$) on the congested
clique are possible?}

We provide a positive answer to this question based on the randomized
technique of dimension reduction in Hamming spaces from \cite{KOR00}. 
With high probability (w.h.p.), our protocol constructs an $O(1)$-factor approximation of an MST of $n$
points in the Hamming space $\{ 0,\ 1\}^n$  using $O(\log^3 n)$ rounds
on the congested clique with $n$ nodes  (see Theorem \ref{theo: hmst} in Section  3).
We also observe that an exact MST of a set of $n$ points in
the space $\{ 0,\ 1\}^n$ can be deterministically constructed in  $O(n^{1-2/\omega +\epsilon
})$ rounds on the congested clique, for any positive constant
$\epsilon $ (Theorem \ref{theo: exact} in Section 3).

Equipped with our protocol for the approximate MST in the Hamming
space $\{ 0,\ 1\}^n$, we can provide our protocol for the Boolean
matrix product of two $n\times n$ matrices on the congested clique.
W.h.p., it uses $\tilde{O}\left(\sqrt {\frac M n+1}\right)$ rounds on
the congested clique with $n$ nodes, where $M$ is the minimum of the
cost of an MST of the rows of the first input matrix and the cost of
an MST of the columns of the second input matrix in the Hamming space
$\{0,1\}^n.$ For $i=1,\dots,n,$ the $i$-th rows of both input matrices
are initially held in the $i$-th node and at the end the $i$-th node
holds the $i$-th row of the product matrix (see Theorem \ref{theo:
  main} in Section 4).  Our protocol for the Boolean matrix product is
superior over that $O(n^{1-2/\omega+\epsilon })$-round protocol from
\cite{K_21} when the rows of the first input matrix or the columns of
the second input matrix are highly clustered so $M$ is substantially
smaller than $n^{3-\frac 4 {\omega}}$ (which is roughly $n^{1.31}$ in
terms of the recent bound on $\omega$ from \cite{VXZ24}).

Our paper is structured as follows.  The next section contains basic
definitions, facts, and two lemmata on routing in the congested clique
model.  Section 3 presents our results on MST in the Hamming space $\{
0,1\}^n$ while Section 4 presents our results on the Boolean matrix
product for highly congested data on the congested clique.

\section{Preliminaries}

For a positive integer $r$, $[r]$ stands for the set of positive integers not exceeding $r.$

The transpose of a matrix $D$ is denoted by $D^t.$

The symbol $\omega$ denotes the smallest real number such that two $n\times n$
matrices can be multiplied using $O(n^{\omega +\epsilon})$
operations over the field of reals, for all $\epsilon > 0.$

The {\em Hamming distance} between two points $p,\ q$ (vectors) in $\{
0,\ 1\}^n$ is the number of the coordinates in which the two points
differ.  Alternatively, it can be defined as the distance
between $p$ and $q$ in
the $\ell_1$ metric over $\{0,\ 1\}^n.$
It is denoted by $h(p,q).$ A {\em witness of the Hamming
  distance} between the two points is the number (index) of a
coordinate in which the points differ.  Note that if $h(p,q)=d$ then
there are $d$ witnesses of the Hamming distance between $p$ and $q.$
\junk{
The {\em communication} of a protocol  within given $t$ rounds on the congested
clique is the total number of messages exchanged by the clique nodes
within these rounds.
The {\em work} done by a protocol within given $t$ rounds on the congested
clique is the total time taken by the local computations at the nodes
of the clique within the $t$ rounds in the logarithmic-cost Random
Access Machine model \cite{AHU}. Specifically, we assume that posting $m$ messages
or receiving  $m$ messages in a round requires $\Omega (m \log N)$
work, where $N$ is the input size. Therefore, the work is always greater
than the communication at least by a logarithmic factor.}

\subsection{Routing}

Lenzen gave an efficient solution to the following fundamental routing problem
in the congested clique model, known as
     the {\em Information Distribution Task} (IDT) \cite{L}:\\
Each node of the congested $n$-clique holds a set of exactly $n$
$O(\log n)$-bit messages with their destinations, with multiple messages from
the same source node to the same destination node allowed.
Initially, the destination of each message is known only to its source node.
Each node is the
destination of exactly $n$ of the aforementioned messages. The
messages are globally lexicographically ordered by their source node, their
destination, and their number within the source node.  For simplicity,
each such message explicitly contains these values, in particular
making them distinguishable. The goal is to deliver all messages to their
destinations, minimizing the total number of rounds.

Lenzen proved that   IDT can be solved in $O(1)$ rounds (Theorem 3.7 in
\cite{L}). He also noted that the relaxed IDT, where each node is required
to send and receive at most $n$ messages, reduces to IDT
in $O(1)$ rounds.
Hence, the following fact holds.

\begin{fact}\cite{L}\label{fact: routing}
  The relaxed Information Distribution
  Task can be deterministically solved within $O(1)$ rounds.
  \end{fact}

The following straightforward generalization of Fact \ref{fact: routing} will be useful.

\begin{lemma}\label{lem: 1}
  Consider a message distribution task on the congested clique,
  where each node has to send at most $kn$ messages and to receive
  at most $\ell n$ messages. The task can be implemented in
  $O(k\ell )$ rounds.
\end{lemma}
\begin{proof}
  Consider the first sub-task, where each node has to send only its
  first $n$ messages. Still some nodes can receive up to $\ell n$
  messages in the sub-task. So, let us consider a sub-sub-task,
  where each node receives only the first $n$ messages that it
  should receive in the sub-task.  To specify the sub-sub-task,
  each node informs each recipent about the numer of messages
  it can deliver in the sub-sub-task. In the next round,
  the recipents inform
  the senders about the number of messages they can receive
  in the sub-sub-task. The messages are lexigraphically ordered
  as described in Preliminaries.
  We can apply the routing protocol
  from Fact \ref{fact: routing}
  in order to implement the sub-sub-task in $O(1)$ rounds.
  After that, each node needs to receive at most $(\ell -1)n$
  messages of the first sub-task. So, we can form and implement a new
  sub-sub-task to decrease the number of messages to be delivered
  to each node to at most $(\ell-2)n$ in $O(1)$ rounds and so
  on. It follows that we can implement the first sub-task in
  $O(\ell ))$ rounds. Then, we can consider the second, third ...
  and $k$-th sub-task analogously defined, and implement
  each of them in $O(\ell)$ rounds analogously.
  The upper bound $O(k\ell )$ on the number of rounds follows.
  \qed
\end{proof}

The following upper bound on the number of rounds will be also
useful.

\begin{lemma}\label{lem: 2}
  Consider a message distribution task on the congested clique,
  where each of the $n$ nodes
  has to send its vector of at most $n$ messages
  to a number of nodes (so each of them receives the same whole
  vector from the node) and each node receives
  the vectors from at most $\ell$ nodes.
  The task can be implemented in
  $O(\ell \log n)$ rounds.
\end{lemma}
\begin{proof}
  We can easily decompose the distribution task into $\ell$ sub-tasks,
  where in the $i$-th sub-task, $i=1,\dots, \ell,$ each node 
  receives only the $i$-th vector among the at most $\ell $ vectors
  it should receive in the distribution task.
  Hence, it is sufficient to show that one can implement such
  a sub-task in $O(\log n)$ rounds. Note that for any two nodes
  $i$ and $j,$ the set of recipients of the vector from
  the node $i$ is disjoint from the set of recipients of the vector
  from the node $j.$

 To implement the sub-task, we shall adhere to the known method of
 doubling in parallel processing. Thus, if a node has more than two
 recipients, they are divided into groups. The first group consists of
 two nodes, the second possibly of four nodes etc.  Each recipient in
 an intermediate group gets also the task of passing further the
 vector to at most two members in the next group.

 To make the nodes aware to which group of recipients of which node
 they belong and to which nodes they should pass the vector further,
 at the beginning, each node sends a message to each of its recipients
 informing that they are its recipients. Next, each node sends a message
 to all other nodes informing them about which node it is an recipient
 of. On the basis of this information, each node finds out to which group
 it belongs to and to which nodes in the next group should pass
 the obtained vector.

   Thus, the sub-task can be implemented in at most $\lceil \log n\rceil $ phases,
   each consisting of $2\times O(1)$ rounds by Lemma \ref{lem: 1}. The upper
   bound of $O(\ell \log n )$ rounds follows.
    \qed
 \end{proof} 

\section{Exact and Approximate Minimum Spanning Tree in $\{ 0,1\}^n$}

First, we shall observe that an exact minimum spanning tree of a set
of $n$ points in the Hamming space $\{ 0,1\}^n$ can be
deterministically computed in
$O(n^{0.158})$ rounds by a composition of the following facts.

\begin{fact}\label{fact: alg}(see Theorem 1 in \cite{CK19})
  The product of two $0-1$ $n\times n$ matrices over the ring
  of integers can be computed on the congested clique
  with $n$ nodes in $O(n^{1-2/\omega +\epsilon })$  rounds,
  for any constant $\epsilon > 0.$
\end{fact}

\begin{fact}\label{fact: now}(see Corollary 2.4  in \cite{K_21})
  The minimum spanning tree of an edge weighted graph
  with $n$ vertices can be deterministically computed
  in $O(1)$ rounds on the congested
  clique with $n$ nodes.
\end{fact}

\begin{theorem}\label{theo: exact}
  A minimum spanning tree of a set of $n$ points in the Hamming space
  $\{ 0,1\}^n$ can be deterministically computed  on the congested
  clique with $n$ nodes  in $O(n^{1-2/\omega +\epsilon })$  rounds,
  for any constant $\epsilon > 0$ ( under the assumption that for
  $i=1,\dots,n,$ the $i$-th node holds the $i$-th input point).
\end{theorem}
\begin{proof}
  First, we shall compute the Hamming distances between all pairs
  of input points. Let $P$ be the $0-1$ $n\times n$ matrix,
  where for $i=1,...,n,$ the $i$-th row of $P$ is
  the $i$-th input point (vector) $p_i$ held in
  the $i$-th node of the clique. By using the $O(n^{1-2/\omega +\epsilon })$-round
  protocol from
  Fact \ref{fact: alg}, we compute the product $C=PP^t$
  so for $i=1,...,n,$ the $i$-th node gets the $i$-th row
  of the product $C.$ Next, let $\bar{P}$ be the $0-1$
  $n\times n$ matrix obtained from $P$ by flipping
  the values of the entries in $P,$ i.e., ones become zeros
  and {\em vice versa}. In a similar manner, we compute
  the product $\bar{C}=\bar{P}\bar{P^t}$ by using
  the protocol from Fact \ref{fact: alg}.
  Note that for $i,\ j \in [n],$ the Hamming distance
  distance between the points $p_i$ and $p_j$ is
  equal to $n$ minus the number of coordinates these
  points share ones and minus the number of coordinates
  they share zeros, i.e., $n-c_{ij}-\bar{c}_{ij}.$
  In result, each node can compute the Hamming
  distances between its input
  point and all other input points so the $O(1)$-round protocol
  for an MST from Fact \ref{fact: now} can be applied.
  \qed
\end{proof}
  
The following fact and remark enable
an efficient randomized dimension reduction 
for the purpose of computing approximate nearest
neighbors and approximate minimum spanning tree for
a set $n$ input points (vectors) in $\{0,\ 1\}^n$
on the congested clique.

\begin{fact}\label{fact: R}\cite{KOR00}
Fix the error parameter $\epsilon \in (0, 1/2),$ dimensions $k, d \ge 1,$
and scale $r\in [1, d].$
For any $k \ge 1,$ there exists a randomized map $f : \{0, 1\}^d \rightarrow \{0, 1\}^k$ and an absolute
constant $C > 0,$
satisfying the following inequalities for any fixed $x, y \in \{0, 1\}^d :$
\begin{description}
\item
if $h(x,y)\le r,$ then $Pr[h(f (x),f (y))\le k/2] \ge  1-e^{C\epsilon^2 k}$;
\item
if  $h(x,y)\ge (1 + \epsilon)r,$  then $Pr[h(f (x),f (y))> (1+\epsilon/2)k/2]\ge  1-e^{C\epsilon^2 k}$.
\end{description}
\end{fact}

\begin{remark}\cite{AIR18}\label{rem: 1}
  The map f can be constructed via a random projection over
  $GF (2)$ by taking $f (x) = Ax,$
  where $A$ is a $k\times d$ matrix for $k = O(\log n/\epsilon^2)$,
  with each entry being $1$ with some fixed probability $\delta$
and $0$ otherwise. The probability $\delta$ depends solely on $r.$
\end{remark}

The following protocol constructs an approximate minimum
spanning tree (MST) of  a set of $n$ points (vectors) in
the Hamming space $\{0,1\}^n$ on the congested clique.
It relies on Fact \ref{fact: R} and Remark \ref{rem: 1}.
 Each node of the clique sends a logarithmic number
of random projections of the point assigned to it into an $O(\log
n)$-dimensional space to the first node. On the
basis of these projections, the latter node constructs locally an
approximate MST of the $n$ input points, where the $i$-th point is
represented by $i,$ for $i=1,\dots ,n$.

\par
   \vskip 5pt
    \noindent
        {\bf protocol} $HMST (S)$
        \par
        \noindent
            {\em Input}: A set $S$ of $n$ points (vectors) in $\{0,1\}^n,$
            each node $i$ holds a distinct input point $p_i$.
            \par
            \noindent
                {\em Output}: An approximate MST of $S$
                held in the first node, where the vertices
                of the tree are represented by the indices of the corresponding points.
                \begin{enumerate}

                \item
                  The first node picks enough large $k=O(\log n /\epsilon^2)$,
                  for $\epsilon$ close to $\frac 12 $, and 
                  generates the random $k\times n$     $0-1$ matrices 
                $A_r$ for $r=1,\ 2, 4,,\dots, 2^{\lceil \log n \rceil} ,$
                  defining the functions $f_r$ by $f_r(x)=A_rx$
                  (see Remark \ref{rem: 1}).
                Next, it sends the matrices to all other nodes.
                \item
                Each node $i$ computes the values of $f_rs$ for its point
                $p_i,$ i.e., it computes $A_rp_i$ for $r=1, 2, 4,\dots, 2^{\lceil \log n \rceil}$,
                and sends the results to the first node.
                \item
                The first node estimates $h(p_i,p_j)$ for each pair of input points
                $p_i$ and $p_j$ as follows. It picks the smallest $r\in \{1,2,\dots,
                 2^{\lceil \log n \rceil}$ \}such that $h(f_r(p_i),f_r(p_j))\le k/2$.
                 The first node estimates $h(p_i,p_j)$ by $r$ via setting
                 $w_{ij} =r.$
                 Next, it finds a minimum
                 spanning tree $T$ in the auxiliary graph on $\{1,2,3,\dots,n\}$,
                 where the weight of the edges $(i,j)$ is $w_{ij}.$
                 The output approximate minimum spanning tree of $S$ is set to $T$,
                 where the vertex $i$ represents the input point $p_i,$ for
                 $i=1,\dots, n.$
                 \junk{ The first node computes an approximate MST of $S$ by using
                the obtained values $f_r(p_i)$ and for example following
                Boruvka's algorithm. Thus, it starts from singleton
                components, forming an edge free  forest $F$, and runs $O(\log n)$ phases. At the beginning of
                each phase, for each current component, it searches for
                an approximate shortest edge to another component as follows.
                For $r=1,2,\dots ,$ it checks if there is an edge $(u,v)$
                connecting the component with another one such that
                $||f_r(u)-f_r(v)||_1\le r.$ If so then it picks such a connection
                minimizing $||f_r(u)-f_r(v)||_1.$ The phase ends by adding
                the picked connections to $F$
                and merging newly connected
                components. The output tree is set to $F.$}
                \end{enumerate}
\begin{lemma}\label{lem: ineq}
Assume the notation from the protocol $HMST(S).$ For each pair
of input points $p_i,\ p_j,$
the following inequalities  hold  w.h.p.:

$$w_{ij}/2\le h(p_i,p_j)+1$$
$$ h(p_i,p_j)\le 1.5w_{ij}.$$
\end{lemma}
\begin{proof}
Suppose tat $w_{ij}\ge 2$ and $h(p_i,p_j)< w_{ij]}/2$ holds.
Then, the inequality $h(f_r(p_i),f_r(p_j))\le k/2$
should hold already for $r=w_{ij}/2$ w.h.p. by the first inequality in 
Fact \ref{fact: R}
with $\epsilon$ close to $\frac 12$ and  sufficiently large $k=O(\log n/ \epsilon^2).$
We obtain a contradiction with the choice of $w_{ij}$ w.h.p.
If $w_{ij} < 2$ then$w_{ij}/2\le h(p_i,p_j)+1$ trivially holds. 
Suppose in turn that $h(p_i,p_j)> 1.5w_{i j}$ holds.
Then, $h(f_r(p_i),f_r(p_j))\ge (1+\frac 15)k/2$ holds
for $r=w_{ij}$ w.h.p.
by the second inequality in Fact \ref{fact: R} with  $\epsilon$
close to $\frac 12$ and $k=O(\log n/ \epsilon^2).$
We obtain again a contradiction with the choice of $w_{ij}$ w.h.p.
\qed
\end{proof}

\begin{lemma}\label{lem: approx}
Let $T$ be the approximate minimum spanning tree
of $S$ constructed by the protocol $HMST(S).$
Consider a minimum spanning tree $U$ of $S$ in
the Hamming metric. Then, the following inequality holds w.h.p.:

$$\sum_{(i,j)\in T}h(p_i,p_j)\le 3 \sum_{(i,j)\in U} h(p_i,p_j) +3n-3$$
\end{lemma}
\begin{proof}
By Lemma \ref{lem: ineq} and the optimality of the MST tree $T$ in the graph on $[n]$ with the
weights $w_{ij}$, we obtain the following chain of inequalities w.h.p.
$$\sum_{(i,j)\in T}h(p_i,p_j)\le
1.5\sum_{(i,j)
  \in T}  w_{ij} $$
$$\le 1.5\sum_
{(i,j)\in U}  w_{ij}\le 3\sum_{(i,j)\in U} h(p_i,p_j) +3n-3$$.
\qed
\end{proof}

\begin{lemma}\label{lem: hmstrounds}
  The protocol $HMST(S)$ can be implemented in $O(\log^3 n)$
  rounds on the congested clique with $n$ nodes.
\end{lemma}
\begin{proof}
  In Step 1, the first node has to send $O(k \log n )$ vectors
  of $n$ messages to each other node. By applying Lemma \ref{lem: 2}
  $O(k \log n)$ times, this can be done in $O(k\log^2 n)$ rounds, i.e.,
  $O(\log^3 n)$ rounds by the choice of $k=O(\log n).$
  In Step 2, each node has to send $O(k \log n)$ messages to the first
  node.  This can be easily done in $O(k\log n)=O(\log^2 n)$ rounds.
    Step 3 is done locally.
  \qed
\end{proof}
\junk{
\begin{lemma}\label{lem: hmstwork}
                    The upper bound on the number of rounds necessary to
                    implement the protocol $HMST$
                    specified in Lemma \ref{lem: hmstrounds} requires
                    $\tilde{O}(n^2)$ work.
                  \end{lemma}
                  \begin{proof}
                    It is easily seen that each node in each of the three steps
                    of the protocol (inclusive the application
                    of Lemma \ref{lem: 2} and indirectly Fact
                    \ref{fact: routing}), with the exception of 
                    the first node in the third
                    step, needs to perform solely $\tilde{O}(n)$ work
                    within the upper bounds on the number of rounds
                    described in the proof of Lemma \ref{lem: hmstrounds}.
                    It remains to observe that the first node can locally
                    estimate the distances  $h(p_i,p_j)$ and construct the
                    minimum spanning tree $T$in the third step
                    in $\tilde{O}(n^2)$sequential time \cite{AHU}.
                    \qed
                  \end{proof}}
                 
Lemmata \ref{lem: approx}, \ref{lem: hmstrounds}
yield our main result in this section.

\begin{theorem}\label{theo: hmst}
  The protocol $HMST(S)$ returns a spanning tree of $S$
  that w.h.p. approximates the minimum spanning tree of $S$
  in the Hamming space $\{0,1\}^n$ within a multiplicative
  factor $3$
  and an additive factor $3n-3$ in
  $O(\log^3 n)$
  rounds
  on the congested clique with $n$ nodes.
\end{theorem}
\begin{proof}
  The weighted graph constructed by the first node is well defined
  so the minimum spanning tree of the graph induces a spanning tree
  of $S.$ The induced spanning tree approximates the minimum
  spanning tree of $S$ in $\{0,1\}^n$ as stated in the theorem
  by Lemma \ref{lem: approx} . The upper bound on
  the number of rounds follows
  from Lemma \ref{lem: hmstrounds}.
\qed
\end{proof}

\section{Boolean Matrix Multiplication for Clustered Matrices on Congested Clique}
Our basic  protocol for the Boolean matrix product of two $n\times n$ Boolean
matrices on the congested clique is intended to provide a good upper
bound on the required number of rounds when the rows of the first matrix $A$
are highly clustered. It relies on the protocol $HMST$ for an
approximate minimum spanning tree in $\{0,\ 1\}^n$ applied to the rows
of the matrix $A.$ Thus, each node of the clique sends a logarithmic number
of random projections of the row assigned to it into an $O(\log
n)$-dimensional space (see Fact \ref{fact: R} and Remark \ref{rem: 1}) to the first node. On the
basis of these projections, the latter node constructs locally an
approximate MST $T$ of the rows of $A$, where the $i$-th row is
represented by $i,$ for $i=1,\dots ,n$ (see Theorem \ref{theo: hmst}).
The first node sends the tree $T$ to all other nodes.
Then, each node determines  the same
traversal $T'$ of $T$. Subsequently, each node
finds out the same decomposition of $T'$ into an appropriate number of
blocks of a balanced Hamming cost (i.e., the sum of Hamming
distances between the endpoints of edges in the traversal block). It
also determines the same decomposition of the second matrix $B$ into blocks of
almost the same number of consecutive columns of $B.$ Furthermore,
each node finds out a single pair consisting of a block of $T'$
and a block of $B$ that will be assigned to it. Then,
each node receives the row of $A$ corresponding to the start vertex
of the traversal block in its pair and the columns of $B$ belonging to
the block of $B$ in its pair.  After that, the node reconstructs the
rows of $A$ corresponding to the vertices in the traversal block and
computes their inner products with all the columns in the block of
$B.$ Finally, the resulting values of the corresponding entries of the
Boolean matrix product of $A$ and $B$ are send to appropriate nodes so the
$i$-th node holds the $i$-th row of the product matrix.

   \par
   \vskip 5pt
    \noindent
        {\bf protocol} $CLUS-MAT(A,B)$
        \par
        \noindent
            {\em Input}: Two $n\times n$ Boolean matrices $A,\ B$,
            each node holds a batch of $2n$ distinct entries
            of $A$ or $B.$
            \par
            \noindent
                {\em Output}:The Boolean matrix product $C$ of $A$ and $ B,$
                the $i$-th node holds the $i$-th row of $C.$

                \begin{enumerate}
               \item
                 The input is rearranged so for $i=1,...,n,$ the $i$-th node holds the $i$-th row
                 of $A$
                 (denoted by $a_i$)
                 and the $i$-th column of $B$.
               \item
                 The protocol $HMST(A)$ is called in order to construct
                 an approximate minimum spanning tree $T$ of the rows
                 of $A$ in the Hamming space $\{0,1\}^n$ in the first node.
                 In the tree $T,$ the points $p_i$ are represented by the
                 indices $i.$
               \item
                 The first node sends the tree $T$ to all other nodes.
              \item
                For $i=1,\dots , n,$ the $i$-th node sends the $i$-th
                row of $A$ to each node $j$, where one of the
                endpoints of the $j$-th edge of $T$ is $i.$
              \item
                For $i=1,\cdot, n-1,$ if the $i$-th edge of $T$
                is $(k,\ell),$ the $i$-th node determines locally
                all the witnesses of the Hamming distance between
                the $k$-th and $l$-th row of $A$.
                Next, it sends the value of the Hamming distance
                between these two rows to all the other nodes.
              \item
                Each node determines the same traversal $T'$ of
                $T$ (having $< 2n$ directed edges).  Next, it 
                computes the cost $M$ of
                the traversal $T'$ in
                the Hamming space $\{0,1\}^n,$ i.e., $M=\sum_{(i,j)\in
                  T'}h(a_i,a_j)$. Then, on the basis of $n$
                and $M$, each node determines the same decomposition of the traversal
                $T'$ into $t=\lceil \sqrt {\frac M n +1}\rceil $ blocks of consecutive
                edges, each block of of total cost $O(M/t).$  It also determines
                the same decomposition of the matrix $B$ into $\lfloor n/ t \rfloor$ blocks
                of consecutive at most $t+1$ columns of $B.$
                Furthermore, each node determines the same assignment
                of each of the at most $t\times \lfloor n/ t \rfloor$
                pairs consisting of a block of $T'$ and a block of $B$
                to distinct nodes. 
              \item
                For $i=1,\cdot, n,$ the $i$-th node sends the $i$-th row
                of $A$ to each node assigned a pair of blocks, where
                the block of the traversal $T'$ begins with
                the $i$-th row of $A$.
              \item
                For $i=1,\cdot, n,$ the $i$-th node
                sends the witnesses of the Hamming distance between
                the rows of $A$ corresponding to the endpoints of
                the $i$-th edge of $T$ to all nodes assigned a pair
                of blocks, where the block of traversal of $T'$
                includes the $i$-th edge.
                \item
                 For $i=1,\cdot, n,$ the $i$-th node sends the $i$-th column
                of $B$ to each node assigned a pair of blocks, where
                the block of $B$ includes
                the $i$-th column of $B.$
              \item
                Each node locally reconstructs the rows of $A$ corresponding
                to the  block of traversal of $T'$ in the pair of blocks
                assigned to it. Next, it computes the entries of $C$
                resulting from the inner products of the reconstructed
                rows of $A$ and the columns of $B$ from the second block
                in the pair assigned to it.
                Finally, it sends each computed entry $c_{ij}$ of $C$
                to the $i$-th node.
               \end{enumerate}

                The following sequence of lemmata provides upper bounds on the number
                of rounds needed to implement the consecutive steps of the protocol
                $CLUS-MAT(A,B).$ In these lemmata, we assume the notation from
                this protocol. In particular, $M$ stands for
                $\sum_{(i,j)\in T'}h(a_i,a_j)$, where $T'$ the traversal of
                the approximate MST of the rows of $A$
                and $a_m$ denotes the $m$-th row of $A$ for $m=1,\dots,n.$

                \begin{lemma}\label{lem: step1}
                    Steps 1,3, 5 and 6 in the protocol $CLUS-MAT$ can be implemented in $O(1)$ rounds.
                  \end{lemma}
                  \begin{proof}
                    Step 1 and 5  can be implemented in $O(1)$ rounds by using the routing protocol from
                    Fact \ref{fact: routing} or Lemma \ref{lem: 1}.
                    To implement Step 3, the first node sends the $i$-th edge
                    of $T$ to the $i$-th node and in the next round the $i$-th
                    node sends it to the other nodes, for $i=1,\dots,n.$
                    Finally, Step 6 is done locally in a single round.
                    \qed
                  \end{proof}

                  By Theorem \ref{theo: hmst}, we obtain the next lemma.

                  \begin{lemma}\label{lem: step2}
                    Steps 2 in the protocol $CLUS-MAT$ can be implemented in $O(\log^3 n)$ rounds.
                  \end{lemma}

                  \begin{lemma}
                  Steps 4 and 7 in the protocol $CLUS-MAT$ can be implemented in $O(\log n)$ rounds.
                  \end{lemma}
                  \begin{proof}
                    In Step 4, each node sends its whole row of $A$ to a number
                    of recipient nodes and needs to receive at most two rows of $A.$
                    Hence, this step can be implemented in $O(\log n)$ rounds
                    by Lemma \ref{lem: 2}.

                    Let us consider Step 7.
                    There are $t=\lceil \sqrt {\frac M n+1}\rceil $ blocks of the traversal
                    $T'$, each of them occurs in $\lfloor \frac n t \rfloor $ different pairs
                      of blocks of $T'$ and $B.$ Thus, each of the at most $t$ nodes
                      whose row of $A$ corresponds to the start vertex
                      of a block of $T'$ has to send its row to all
                      the nodes assigned a block pair, where the traversal
                      block starts from the vertex corresponding to the row.
                      Thus, at most $t \times n/t$ copies  of the
                      $t$ rows have to be delivered.
                      Each node is assigned only
                      a single block pair, so it needs to receive only a single
                      copy of a row of $A.$ Hence, Step 7 can also be implemented  in
                      $O(\log n)$ rounds by Lemma \ref{lem: 2}.
                      \junk{
                      The problem is that some nodes
                      need to send their rows of $A$ to at least $n/t$ nodes,
                      potentially even more, if the same vertex starts several
                      blocks of $T'.$ So, we can solve this congestion obstacle
                      analogously as in the case of Step by the doubling technique
                      which results in $O(\log n)$ rounds.}
                      \qed
                  \end{proof}

                  \begin{lemma}
                    Step 8 in the protocol $CLUS-MAT$ can be implemented in\\
                    $O\left(\log n \sqrt {\frac M n+1}\right)$ rounds  on the congested clique.
                  \end{lemma}
                  \begin{proof}
                    It follows from
                    the definition of $M$
                    that there are at most  $M$ witnesses of Hamming distances
                    between the rows of $A$ corresponding to the endpoints
                    of edges in $T'.$ Recall that $t=\lceil \sqrt {\frac M n+1}\rceil .$
                    Each of the witnesses has to be sent to the $\lfloor n/t \rfloor $
                    nodes, where the travel block in the assigned block pair
                    includes the (directed) edge corresponding to the witness.
                    The witness delivery is performed in two stages.

                    In the first stage, for
                    each block of the traversal $T'$, $\lfloor n/t\rfloor $ consecutive 
                    node representatives are designated by the same
                    algorithm run locally at each node. Next, each node for each its witness
                    computes the interval of indices of the representative nodes for the block
                    of $T'$ the witness belongs to. Next, the node picks one of the representatives
                    nodes to send the witness  to such that each representative node receives $O(n)$ such witnesses.
                    The determination of the target representative nodes
                    is done on the basis of the knowledge
                    of the blocks of $T'$ and the Hamming distances between the rows of
                    $A$ corresponding to the endpoints of edges in $T'$ (recall Steps 3 and 5)
                    by using the same assignment algorithm.
                    Thus, in the first stage each node needs to send at most $n$
                    witnesses (i.e., $O(n)$ messages) as well as to receive
                    at most $O(n)$witnesses (i.e.m $O(n)$ messages).
                    By Lemma \ref{lem: 1}, this stage can be implemented in $O(1)$ rounds.

                    In the second stage, each representative node of a
                    block $b$ of $T'$ sends its vector of $O(n)$
                    messages to each of the $\lfloor n/t \rfloor $ nodes assigned a
                    block pair including the block $b.$ 
                    The second stage can be easily decomposed in $O(1)$ sub-stages
                    such that in each consecutive sub-stage, each node
                    representative of a block of $T'$ sends its
                    consecutive sub-vector of at most $n$ messages to
                    the $\lfloor n/t\rfloor $ nodes. In each sub-stage, each node
                    receives $O(\frac M {nt})=O(t)$ vectors of $n$ or fewer
                    messages.It follows from Lemma \ref{lem: 2} that each sub-stage
                    can be implemented in $O(t\log n )$ rounds. 
                    \junk{
                    Since some of the nodes can have up to linear in $n$ number of witnesses,
                    in order to avoid a congestion, each of the nodes locally determines
                    the same balancing assignment, so the nodes having more than the average
                    number of witness pass some of their witnesses to the nodes having
                    a lower number of witnesses. The averaging phase can be implemented
                    in $O(1)$ rounds by Fact . In result, each node has $O(M/n)$ witnesses
                    to distribute, each of them should be delivered to $n/t$ nodes,
                    so each node should send $O(\frac M {nt}n=M/t)$ messages and receive up
                    to $$O(\frac M {nt}n=M/t)$M/t$ messages.

                    Since $M/t^2=n,$ this process can be easily decomposed into $O(t)$
                    phases, so in each phase each node has to send and receive up
                    to $n$ messages, so each of the phases can be implemented in
                    $O(1)$ rounds by Fact .}
                    \qed
                  \end{proof}

                  \begin{lemma}
                    Step 9 in the protocol $CLUS-MAT$ can be implemented in\\
                    $O
                    \left( \log n \sqrt {\frac M n+1}\right)$ rounds.
                  \end{lemma}
                  \begin{proof}
                    There are $t$ pairs of blocks, where the block of $B$ includes
                    the $i$-th column of $B.$
                    Thus, the $i$-th node has to sens its column to
                    $t=O\left(\sqrt{\frac M n+1}\right)$ nodes.
                    On the other hand, each node in order to construct
                    the block pair assigned to it needs to receive at most $t+1$
                    consecutive columns of $B,$  each from a distinct node.
                    This can be done in $O\left(\log n \sqrt {\frac M n+1}\right)$ rounds by Lemma \ref{lem: 2}.
                    \qed
                  \end{proof}

                  \begin{lemma}\label{lem: step10}
                    Step 10 in the protocol $CLUS-MAT$ can be implemented in
                    $O\left(\sqrt {\frac M n+1}\right)$ rounds.
                  \end{lemma}
                  \begin{proof}
                    After the local computations, each node has to send the at most
                    $n \times t=n \times O\left(\sqrt \frac M n+1\right)$ computed
                    entries of $C$ to the nodes assigned
                    to the rows of $C$ the entries belong to, respectively.
                    On the other hand, each node is recipient of $O(n)$ entries
                    (some entries can be computed and delivered twice because
                    some undirected edges of $T$ can appear twice as directed ones in the traversal $T'$).
                    Hence, the process can be implemented in $O(t)$ rounds
                    by Lemma \ref{lem: 1}.
                    \qed
                  \end{proof}

                     \begin{theorem}\label{theo: main}
                    Let $A,\ B$ be two $n\times n$ Boolean matrices,
                    whose $i$-th rows, $i=1,\dots,n,$  are initially held in the $
                    i$-th node of the congested clique
                    on $n$ nodes. Let $M$ be the minimum of the  cost of an MST
                    of the rows of $A$ and the cost of an MST of the columns
                    of $B$ in the Hamming space $\{0,1\}^n.$
                    W.h.p. the Boolean matrix product of $A$ and $B$ can be computed
                    in $\tilde{O}\left(\sqrt {\frac M n+1}\right)$ rounds
                                        on the congested clique with $n$ nodes such that the $i$-th row of the product
                    is provided in the $i$-th node of the congested clique
                     \end{theorem}
                     \begin{proof}
                       Let $M_A,\ M_B$ be the cost of an MST of the rows of $A$
                       and the cost of an MST of the columns of $B$ in
                       the Hamming space $\{0,1\}^n$, respectively.
                       By Theorem \ref{theo: hmst}
                       and Lemmata \ref{lem: step1}-\ref{lem: step10}, the protocol
                       $CLUS-MAT(A,B)$ computes the Boolean matrix
                       product of $A$ and $B$ in $\tilde{O}\left(\sqrt {\frac {M_A} n+1}\right)$ rounds w.h.p.
                       We can also apply
                       the protocol $CLUS-MAT(B^t,A^t)$ in order compute
                       the transpose of the product
                       of $A$ and $B$
                       in $\tilde{O}\left(\sqrt {\frac {M_B} n+1}\right)$ rounds w.h.p.
                       So, we can first compute approximations of both minimum spanning trees in
                       $O(\log^3 n)$ rounds
                       by Theorem \ref{theo: hmst}
                       and then run the protocol on $A$ and $B,$ or $B^t$ and $A^t$, depending
                       on which of the approximations has a smaller cost.
                       \qed.
                     \end{proof}

                     \section{Final Remarks}
                     In the sequential Random Access Machine model,
                     highly clustered rows of the first matrix or columns
                     of the second matrix allow for a more efficient computation
                     of the Boolean matrix product \cite{BL,GL03} while
                     in the congested clique model the clustering enables 
                     for fewer rounds necessary to distribute fragments of the input
                     matures among the nodes (Theorem \ref{theo: main}).
                     For this reason, our protocol for the Boolean matrix product
                     cannot be regarded as only  an implementation of the
                     sequential algorithm from \cite{BL}  on the congested
                     clique.

                     The extended Hamming distance ($eh(\ ,\ )$)
                     introduced in \cite{GL03} is a generalization of
                     the Hamming distance to include blocks of zeros
                     and ones. It is defined recursively for two $0-1$
                     strings (alternatively, vectors or sequences)
                     $s=s_1 s_2\dots s_m$ and $u=u_1 u_2\dots u_m$ as
                     follows:
                     $$eh(s,u)=eh(s_{\ell +1}\dots s_m,u_{\ell +1}\dots u_m)+
                     (s_1+u_1)\mod 2,$$
                     where $\ell$ is the maximum number such that
                     $s_j=s_1$ and $u_j=u_1,$ for $j=1,\dots,\ell.$
                     Note that the extended Hamming distance between
                     any two strings or vectors never exceeds the
                     Hamming one.
                     We can carry over our bound of
                     $\tilde{O}\left(\sqrt {\frac M n+1}\right)$
                     rounds for the Boolean matrix product 
                     to the extended  Hamming distance
                     in $\{0,\ 1\}^n$  analogously as it was
                     done in case of the sequential model in
                     \cite{GL03}.

                     Let us define the work done by a protocol within
                     given $t$ rounds on the congested clique as the
                     total time taken by the local computations at the
                     nodes of the clique, inclusive posting and
                     receiving messages, within the $t$ rounds, say,
                     in the logarithmic-cost Random Access Machine
                     model \cite{AHU}. Our protocols do not overuse
                     the unlimited local computational power in the
                     congested clique model.  Assuming that Lenzen's
                     $O(1)$-round routing protocol involves
                     $\tilde{O}(n^2)$ work, it is not difficult to
                     observe that the total work done by our protocol
                     for the approximate minimum spanning tree of $n$
                     points in $\{0,\ 1\}^n$ is $\tilde{O}(n^2)$ while
                     that done by our protocol for the $n\times n$
                     Boolean matrix product is $\tilde{O}(n(n+M))$
                     w.h.p.  The latter upper bound asymptotically
                     matches the time complexity of the sequential
                     MST-based algorithm for Boolean matrix product
                     from \cite{BL}.

                     
  \small
  \bibliographystyle{abbrv}                    
\bibliography{Voronoi}
                     \vfill
                     \end{document}